\documentclass[10pt,twocolumn]{IEEEtran}%

\usepackage{amsmath}
\usepackage{amsfonts}
\usepackage{amssymb}
\usepackage{graphicx}
\usepackage{psfrag}
\usepackage{cite}
\usepackage{algorithm}
\usepackage{color}
\usepackage{float}
\usepackage{epsfig,array,multicol,verbatim,algorithmic}%
\IEEEoverridecommandlockouts
\newtheorem{condition}{\bf Condition}

\newtheorem{theorem}{\bf Theorem}
\newtheorem{property}{\bf Property}
\newtheorem{proposition}{\bf Proposition}

\newtheorem{definition}{\bf Definition}

\newlength{\aligntop}
\newlength{\alignbot}
\makeatletter
\renewenvironment{align}{%
  \vspace{\aligntop}
  \start@align\@ne\st@rredfalse\m@ne
}{%
  \math@cr \black@\totwidth@
  \egroup
  \ifingather@
    \restorealignstate@
    \egroup
    \nonumber
    \ifnum0=`{\fi\iffalse}\fi
  \else
    $$%
  \fi
  \ignorespacesafterend%
  \vspace{\alignbot}\par\noindent
}

\begin{document}

\title{Hedonic Coalition Formation for Distributed Task Allocation among Wireless Agents}
\author{Walid Saad, Zhu Han, Tamer Ba\c{s}ar, M\'{e}rouane Debbah, and Are Hj{\o}rungnes \vspace{-1.1cm} \thanks{W.~Saad and A. Hj{\o}rungnes are
with the UNIK Graduate University Center, University of Oslo, Oslo, Norway, e-mails: \texttt{\{saad,arehj\}@unik.no}.
Z.~Han is with Electrical and Computer Engineering Department, University of Houston, Houston, Tx, USA, email: \texttt{zhan2@mail.uh.edu}. T. Ba\c{s}ar is with the Coordinated Science Laboratory, University of Illinois at Urbana Champaign, IL, USA, email: \texttt{basar1@illinois.edu}. M.~Debbah is the Alcatel-Lucent chair, SUPELEC, Paris, France e-mail:
\texttt{merouane.debbah@supelec.fr}. This research is supported by the Research Council of Norway
 through the projects 183311/S10, 176773/S10 and 18778/V11, and by a grant from AFOSR.  A preliminary version of this work [35] appeared in the proceedings of the International Conference on Game Theory for Networks (GameNets), May, 2009. }}
\maketitle

\begin{abstract}
 Autonomous wireless agents such as unmanned aerial vehicles, mobile base stations, or self-operating
 wireless nodes present a great potential for deployment in next-generation wireless networks. While
 current literature has been mainly focused on the use of agents within robotics or software
 engineering applications, we propose a novel usage model for self-organizing agents suited to wireless networks. In the proposed model, a number of agents are required  to collect data from several arbitrarily located tasks. Each task represents a queue of packets that  require collection and subsequent wireless transmission by the agents to a central receiver. The
 problem is modeled as a \emph{hedonic coalition formation} game between the agents and the tasks that interact in order to form disjoint coalitions.  Each formed coalition is modeled as a polling  system consisting of a number of agents which move between the different  tasks present in the coalition, collect and transmit the packets.  Within each coalition, some agents can also take the role of a relay for improving the packet success rate of the transmission. The proposed algorithm allows the tasks and the agents to take distributed decisions to join or leave a coalition, based on the achieved benefit in terms of effective throughput, and the cost in terms of delay. As a result of these decisions, the agents and tasks structure themselves into independent disjoint coalitions which constitute a Nash-stable network partition. Moreover, the proposed algorithm allows the agents and tasks to adapt the topology to environmental changes such as the arrival/removal of tasks or the mobility of the tasks. Simulation results show how the proposed algorithm allows the agents and tasks to self-organize into independent coalitions, while improving the performance, in terms of average player (agent or task) payoff, of at least $30.26\%$ (for a network of $5$ agents with up to $25$~tasks) relatively to a scheme that allocates nearby tasks equally among agents.
\end{abstract}
\vspace{-0.1cm}
\section{Introduction}
%Many domains, such as battlefield operations management, disaster
%response and crisis management, and physical infrastructure and
%cyber security monitoring are characterized by dense real-time
%sensing, large numbers of distributed heterogeneous information
%sources, and  a variety of distributed, communicating, and
%network-enabled actors and agents. In these domains, there is
%the need to automatically and continuously identify and act
%on complex, often incomplete and unpredictable dynamic situations.
%As a result, effective methods of situation recognition, prediction,
%reasoning and control are required -- operations collectively
%identifiable as Situation Management.Situation Management intersects with trends in information fusion,
%intelligent sensing, sensing grids, complex event processing
%architectures, and situation-awareness and context-awareness.
%Often situations involve a large number of inter-dependent
%dynamic objects that change their states in time and space,
%and engage each other into fairly complex relations. From a
%management viewpoint it is important to understand the situations
%in which these objects participate, to recognize emerging trends
%and potential threats, and to undertake required actions.
%Understanding of dynamic situations requires complex cognitive
%modeling of situations and continuous sensing, collection, and
%fusion of signal and human intelligence events and reports.
Next generation wireless networks will present a highly complex and dynamic environment characterized by a large number of heterogeneous information sources, and a variety of distributed network nodes. This is mainly due to the recent emergence of large-scale, distributed, and heterogeneous communication systems which are continuously increasing in size, traffic, applications, services, etc. For maintaining a satisfactory operation of such networks, there is a constant need for dynamically optimizing their performance, monitoring their operation and reconfiguring their topology. Due to the ubiquitous nature of such wireless networks, it is inherent to have self-organizing autonomous nodes (agents), that can service these networks at different levels such as data collection, monitoring, optimization, maintenance, and others \cite{MF00,MF03,CR02,AL00,AL01,AL02}. These nodes can be owned by the authority maintaining the network, and must be able to survey large scale networks, and perform very specific tasks at different points in time, in an autonomous manner, with little reliance on any centralized authority \cite{MF00,MF03,CR02,AL00,AL01,AL02}.

While the use of autonomous agents has been well investigated in robotics, computer systems or software engineering, research models tackling the use of such agents in wireless and communication networks are few. However, recently, the need for such agents in wireless networks has become of noticeable importance as many next-generation networks encompass several types of wireless devices, such as cognitive devices or unmanned aerial vehicles~(UAVs), that are autonomous and self-adapting~\cite{MF00,MF03,CR02,AL00,AL01,AL02}. A key challenge in this area is the problem of task allocation among a group of agents that need to execute a number of tasks. This problem has been already tackled in  areas such as robotics control~\cite{UAV04,UAV05,UAV06} or software systems~\cite{UAV07,UAV09}. However, most of the existing models are unsuitable for task allocation in wireless networks due to many reasons: (i)- The task allocation problems studied in existing work are tailored for military operations, computer systems, or software engineering and, thus,  cannot be readily applied in models pertaining to wireless networks, (ii)- the tasks are generally considered as static abstract entities with very simple characteristics and no intelligence (e.g. the tasks are just points in a plane) which is a major limitation, and (iii)- the existing models do not consider any aspects of wireless networks such as the characteristics of the wireless channel, the data traffic, the need for wireless transmission, or other wireless-specific specifications. In this context, numerous applications in next-generation wireless networks require a number of agent-nodes to perform specific wireless-related tasks that emerge over time and are not pre-assigned. One example is the case where a number of wireless nodes are required to monitor the operation of the network or perform relaying at different times and locations \cite{MF00,MF03,AL00,AL01,AL02}. In such applications, the objective is to develop algorithms enabling the agents to autonomously share the tasks among each other. The main existing contributions within wireless networking in this area \cite{UAV00,ZH00,UAV01,UAV02,UAV03}, are focused on deploying UAVs which can act as self-deploying autonomous agents that can efficiently perform \emph{pre-assigned} tasks in applications such as connectivity improvement in ad hoc network \cite{ZH00}, routing \cite{UAV01,UAV02}, and medium access control \cite{UAV03}. However, these contributions focus on \emph{centralized} solutions for specific problems such as finding the optimal locations for the deployment of UAVs or devising efficient routing algorithms in ad hoc networks in the presence of UAVs. In these papers, the tasks that the agents must accomplish are \emph{pre-assigned and pre-determined}. In contrast, many applications in wireless networks require the agents to autonomously assign the tasks among themselves. Hence, it is inherent to devise algorithms, in the context of wireless networks, that allow an autonomous and distributed task allocation process among a number of \emph{wireless agents}\footnote{The term \emph{wireless agent} refers to any node that can act autonomously and can perform wireless transmission. Examples of wireless agents are UAVs \cite{ZH00}, mobile base stations \cite{AL00,AL01,AL02}, cognitive wireless devices \cite{CR02}, or self-deploying mobile relay stations \cite{MF03}.} with little reliance on centralized entities.
 %and to adapt the task allocation to environmental changes such as the removal of tasks, or the additions of new tasks.

The main contribution of this paper is to propose a novel wireless-oriented model for the problem of task allocation among a number of autonomous agents. The proposed model considers a number of wireless agents that are required to collect data from arbitrarily located tasks. Each task represents a source of data, i.e., a queue with a Poisson arrival of packets, that the agents must collect and transmit via a wireless link to a central receiver. This formulation is deemed suitable to model several problems in next-generation networks such as video surveillance in wireless networks, self-deployment of mobile relays in IEEE 802.16j networks \cite{MF03}, data collection in ad hoc  and sensor networks \cite{AL02}, operation of mobile base stations in vehicular ad hoc networks \cite{AL00} and mobile ad hoc networks  \cite{AL01} (the so called \emph{message ferry} operation), wireless monitoring of randomly located sites, autonomous deployment of UAVs in military ad hoc networks, and many other applications. For allocating the tasks, we introduce a novel framework from coalitional game theory, known as \emph{hedonic coalition formation}. Albeit hedonic games have been widely used in game theory, to the best of our knowledge, no existing work utilized this framework in a communication or wireless environment. Thus, we model the task allocation problem as a hedonic coalition formation game between the agents and the tasks, and we introduce an algorithm for forming coalitions. Each formed coalition is modeled as a polling system consisting of a number of agents, designated as \emph{collectors}, which act as a single server that moves continuously between the different tasks (queues) present in the coalition, gathering and transmitting the collected packets to a common receiver. Further, within each coalition, some agents can act as \emph{relays} for improving the packet success rate during the wireless transmission. For forming coalitions, the agents and tasks can autonomously make a decision to join or leave a coalition based on well defined individual preference relations. These preferences are based on a coalitional value function that takes into account the benefits received from servicing a task, in terms of effective throughput (data collected), as well as the cost in terms of the polling system delay incurred from the time needed for servicing all the tasks in a coalition. We study the properties of the proposed algorithm, and show that it always converges to a Nash-stable network partition. Further, we investigate how the network topology can self-adapt to environmental changes such as the deployment of new tasks, the removal of existing tasks, and low mobility of the tasks. Simulation results show how the proposed algorithm allows the network to self-organize, while ensuring a performance improvement, in terms of average player (task or agent) payoff, compared to a scheme that assigns nearby tasks equally among the agents.

The remainder of this paper is organized as follows: Section~\ref{sec:sysmodel} presents and motivates the proposed system model. In Section~\ref{sec:gmodel}, we model the task allocation problem problem as a transferable utility coalitional game and propose a suited utility function. In Section~\ref{sec:gform}, we classify the task allocation coalitional game as a hedonic coalition formation game, we discuss its key properties and we introduce the algorithm for coalition formation. Simulation results are presented, discussed, and analyzed in Section~\ref{sec:sim}. Finally, conclusions are drawn in Section~\ref{sec:conc}.\vspace{-0.3cm}

\section{System Model}\label{sec:sysmodel}\vspace{-0.2cm}
Consider a network having $M$ wireless agents that belong to a single network operator and  that are controlled by a central command center (e.g., a central controller node or a satellite system). These agents are required to service $T$ tasks arbitrarily located in a geographic area that has an associated central wireless receiver connected to the command center. In general, the tasks are entities belonging to one or more independent owners\footnote{The scenario where all tasks and agents are owned by the same entity is a particular case of this generic model.}. The owners of the tasks can be, for example, service providers or third party operators. We denote the set of agents and tasks by $\mathcal{M}=\{1,\ldots,M\}$, and $\mathcal{T}=\{1,\ldots,T\}$, respectively. We consider only the case where the number of tasks is larger than the number of agents, hence, $T > M$. The main motivation behind this consideration is that, for most networks, the number of agents assigned to a specific area is generally small, e.g., due to cost factors. Each task $i \in \mathcal{T}$ represents an M/D/1 queueing system\footnote{Other queue types, e.g., M/M/1, can be considered without loss of generality.}, whereby packets of constant size $B$ are generated using a Poisson arrival with an average arrival rate of $\lambda_i$. Hence, we consider different classes of tasks each having its corresponding $\lambda_i$. The tasks we consider are sources of data that cannot send their information to the central receiver (and, subsequently, to the command center) without the help of an agent. These tasks can represent a group of mobile devices, such as sensors, video surveillance devices, or any other static or dynamic wireless nodes that have limited power and are unable to provide long-distance transmission. Such devices (tasks) need to buffer their data locally and wait to be serviced by an agent that can collect the data. For example, an agent such as a mobile station or a UAV can provide a line-of-sight link to facilitate the transmission from the tasks to the receiver. The tasks can also be mapped to any other source of packet data that require collection by an agent for transmission\footnote{The tasks can also be moving with a periodic low mobility.}. To service a task, each agent is required to move to the task location, collect the data, and transmit it using a wireless link to the central receiver. The command center periodically downloads this data from the receiver, e.g., through a backbone network. Each agent $i \in \mathcal{M}$ offers a link transmission capacity of $\mu_i$, in packets/second, using which the agent can service the tasks' data. The quantity $\frac{1}{\mu_i}$ thus represents the well-known service time for a single packet that is being serviced by agent $i$. The agent which is collecting the data from a task is referred to as \emph{collector}. In addition, each agent $i \in \mathcal{M}$ can transmit the data to the receiver with a maximum transmit power of $P_i=\tilde{P}$, assumed the same for all agents with no loss of generality.

The proposed model allows each task to be serviced by multiple agents, and also, each agent (or group of agents) to service multiple tasks. Whenever a task is serviced by multiple agents, each agent can act as either a \emph{collector} or a \emph{relay}. Any group of agents that act together for data collection from the same task, can be seen as a single collector with improved link transmission capacity. In this paper, we consider that the link transmission capacity depends solely on the capabilities of the agents and not on the nature of the tasks. In this context, given a group of agents $G \subseteq \mathcal{M}$ that are acting as collectors for any task, the total link transmission capacity with which tasks can be serviced with by $G$ can be given by
\begin{equation}\label{eq:srate}
\mu^{G} =  \sum_{j\in G}\mu_j.
\end{equation}

For forming a single collector, multiple agents can easily coordinate the data extraction and transmission from every task, so as to allow a larger link transmission capacity for the serviced task as in (\ref{eq:srate}). The transmission of the packets by the agents from a task $i \in \mathcal{T}$ to the central receiver is subject to packet loss due to the fading on the wireless channel. Thus, in addition to acting as collectors, some agents may act as \emph{relays} for a task. These relay-agents would locate themselves at equal distances from the task (given that the task is already being served by \emph{at least one} collector), and, hence, the collectors transmit the data to the receiver through multi-hop agents, improving the probability of successful transmission. In Rayleigh fading, the probability of successful transmission of a packet of size $B$ bits from the collectors present at a task $i \in \mathcal{T}$ through a path of $m$ agents, $Q_i=\{i_1,\ldots,i_{m}\}$, where $i_1 = i$ is the task being serviced, $i_{m}$ is the central receiver, and any other $i_h \in Q_i$ is a \emph{relay}-agent, is given by
\begin{equation}\label{eq:suc}
\textrm{Pr}_{i,\textrm{CR}} = \prod_{h=1}^{m-1} \textrm{Pr}_{i_h,i_{h+1}}^B,
\end{equation}
where $\textrm{Pr}_{i_h,i_{h+1}}$ is the the probability of successful transmission of a single bit from agent $i_h$ to agent (or the central receiver) $i_{h+1}$. This probability can be given by the probability of maintaining the SNR at the receiver above a target level $\nu_0$ as follows \cite{DM00}
\begin{equation}\label{eq:p}
\textrm{Pr}_{i,i_{h+1}} =\exp{\left(-\frac{\sigma^2\nu_0(D_{i_{h},i_{h+1}})^{\alpha}}{\kappa \tilde{P}}\right)},
\end{equation}
where $\sigma^2$ is the variance of the Gaussian noise, $\kappa$ is a path loss constant, $\alpha$ is the path loss exponent, $D_{i_{h},i_{h+1}}$ is the distance between nodes $i_{h}$ and $i_{h+1}$, and $\tilde{P}$ is the maximum transmit power of agent $i_h$.
\begin{figure}[!t]
\begin{center}
\includegraphics[angle=0,width=90mm]{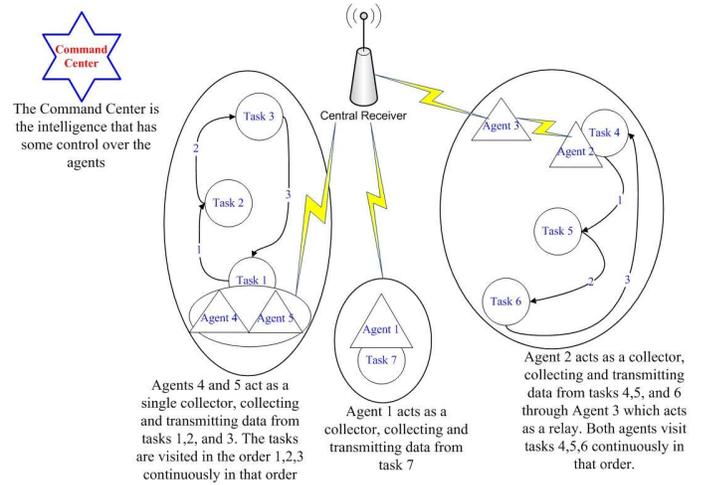}
\end{center}\vspace{-0.4cm}
\caption {An illustrative example of the proposed model for task allocation in wireless networks (the agents are dynamic, i.e., they move from one task to the other continuously).} \label{fig:fir}\vspace{-0.6cm}
\end{figure}

For servicing a number of tasks $C \subseteq \mathcal{T}$, a group of agents $G\subseteq \mathcal{M}$ (collectors and relays) can sequentially move from one task to the other in $C$ with a constant velocity $\eta$. The group $G$ of agents, servicing the tasks in $C$, stop at each task, with the collectors collecting and transmitting the packets using the relays (if any). The collectors would move from one task to the other, only if all the packets in the queue at the current task have been transmitted to the receiver (the process through which the agents move from one task to the other for data collection is cyclic). Simultaneously with the collectors, the relays also move, positioning themselves at equal distances on the line connecting the task being currently served by the collectors, and the central receiver. With this proposed model, the final network will consist of groups of tasks serviced by groups of agents, continuously. An illustration of this model is shown in Fig.~\ref{fig:fir}. Given this proposed model, the main objective is to provide an algorithm for distributing the tasks between the agents, given the operation of the agents previously described and shown in Fig.~\ref{fig:fir}.\vspace{-0.3cm} %In other words, the main goal is to allow the agents and task to autonomously form the coalitional structure in Figure~\ref{fig:fir} and adapt it to environmental changes. For this purpose, the following sections formulate a game theoretic approach for achieving this objective.\vspace{-0.5cm}

\section{Coalitional Game Formulation}\label{sec:gmodel}\vspace{-0.1cm}
%In this section, we model the proposed problem as a coalitional game with transferable utility, and we propose a suitable utility function.\vspace{-0.5cm}
\subsection{Game Formulation}
By inspecting Fig.~\ref{fig:fir}, one can clearly see that the task allocation problem among the agents can be mapped into the problem of the formation of coalitions. In this regard, coalitional game theory \cite[Ch.~9]{Game_theory2} provides a suitable analytical tool for studying the formation of cooperative groups, i.e., coalitions, among a number of players. For the proposed model, the coalitional game is played between the agents and the tasks. Hence, the players set for the proposed task allocation coalitional game is denoted by $\mathcal{N}$, and contains both agents and tasks, i.e., $\mathcal{N} = \mathcal{M} \cup \mathcal{T}$. Hereinafter, we use the term \emph{player} to indicate either a task or an agent.

For any coalition $S \subseteq \mathcal{N}$ containing a number of agents and tasks, the agents belonging to this coalition can structure themselves into collectors and relays. Subsequently, as explained in the previous section, within each coalition, the collector-agents will continuously move from one task to the other, stopping at each task, and transmitting all the packets available in the queue to the central receiver, through the relay-agents (if any). This proposed task servicing scheme can be mapped to a well-known concept that is ubiquitous in computer systems,  which is the concept of a \emph{polling system}\cite{PS00}. In a polling system, a single server moves between multiple queues in order to extract the packets from each queue, in a sequential and cyclic manner. Models pertaining to polling systems have been widely developed in various disciplines ranging from computer systems to communication networks, and different strategies for servicing the queues exist \cite{PS00,PS01,PS02}.  In the proposed model, the collectors of every coalition are considered as a single server that is servicing the tasks (queues) sequentially, in a cyclic manner, i.e., after servicing the last task in a coalition $S \subseteq \mathcal{N}$, the collectors of $S$ return to the first task in $S$ that they previously visited hence repeating their route continuously. Whenever the collectors stop at any task $i \in S$, they collect and transmit the data present at this task until the queue is empty. This method of allowing the server to service a queue until emptying the queue is known as the \emph{exhaustive} strategy for a polling system, which is applied at the level of every coalition $S \subseteq \mathcal{N}$ in our model. Moreover, the time for the server to move from one queue to the other is known as the \emph{switchover} time. Consequently, we highlight the following property:
\begin{property}
In the proposed task allocation model, every coalition $S \subseteq \mathcal{N}$ is a polling system with an \emph{exhaustive} polling strategy and deterministic non-zero \emph{switchover} times. In each such polling system $S$, the collector-agents are seen as the polling system server, and the tasks are the queues that the collector-agents must service.
\end{property}

For any coalition $S$, once the queue at a task $i \in S$ is emptied, the collectors and relays in a coalition move from task $i$ to the next task $j \in S$ with a constant velocity $\eta$, incurring a switchover time $\theta_{i,j}$. The switchover time in our model corresponds to the time it takes for all the agents (collectors and relays) to move from one task to the next. Assuming all agents start their mobility at the same time, this switchover time maps to the time needed for the farthest agent to move from one task to the next. Since we consider only straight line trajectories for collectors and relays, and due to the fact that the relays always position themselves at equal distances on the line connecting the tasks in a coalition to the receiver, we have the following property (clearly seen through the geometry of Fig.~\ref{fig:fir}).
 \begin{property}
Within any given coalition $S$, the switchover time between two tasks corresponds to the \emph{constant} time it takes for one of the \emph{collectors} to move from one of the tasks to the next.
\end{property}

Having modeled every coalition $S \subseteq \mathcal{N}$ as a polling system, we investigate the average delay incurred per coalition. For polling systems, finding exact expressions for the delay at every queue is a highly complicated task and no general closed-form expressions for the delay at every queue in a polling system can be found \cite{PS00,PS01}\footnote{Note that some approximations \cite{PS01} exist for polling systems under heavy traffic or large switchover times, but in our problem, they are not suitable as we require a more general delay expression.}. A key criterion used for the analysis of the delay incurred by a polling system is the \emph{pseudo-conservation} law that provides closed-form expressions for weighted sum of the means of the waiting times at the queues \cite{PS00,PS01}. For providing the pseudo-conservation law for a coalition $S \subseteq \mathcal{N}$ composed of a number of agents and a number of tasks, we make the following definitions. First, within coalition $S$, a group of agents  $G_S \subseteq S\cap\mathcal{M}$ are designated as collectors. Second, for each task $i \in S \cap \mathcal{T}$ with an average arrival rate of $\lambda_i$, and served by a number of collectors $|G_S|$ with a link transmission capacity of $\mu^{G_S}$ (as given by (\ref{eq:srate})), we define the utilization factor of task $i$ $\rho_i = \frac{\lambda_i}{\mu^{G_S}}$. Further, we define $\rho_S \triangleq \sum_{i \in S \cap \mathcal{T}} \rho_i$. Given these definitions, for a coalition $S$, the weighted sum of the means of the waiting times by the agents at all the tasks in the coalition are given by the pseudo-conservation law as follows \cite[Section.~VI-B]{PS01} (taking into account that our switchover and service times are deterministic)
\begin{align}\label{eq:pclaw}
\sum_{i \in S\cap \mathcal{T}} \rho_i \bar{W}_i = \rho_S \frac{\sum_{i\in S\cap \mathcal{T}} \frac{\rho_i}{\mu^{G_S}} }{2(1-\rho_S)} +  \rho_S \frac{\theta_S^{2}}{2}\\\nonumber
 + \frac{\theta_S}{2(1-\rho_S)}\left[\rho_S^2 - \sum_{i\in S\cap \mathcal{T}} \rho_i^2\right],
\end{align}
where $\bar{W}_i$ is the mean waiting time at task $i$ and $\theta_S = \sum_{h=1}^{|S\cap \mathcal{T}|}\theta_{i_h,i_{h+1}}$ is the sum of the switchover times given a path of tasks $\{i_1,\ldots,i_{|S\cap \mathcal{T}|}\}$ followed by the agents, with $i_h \in S \cap \mathcal{T},\ \forall\ h\in \{1,\ldots,|S\cap \mathcal{T}|\}$ and $i_{|S\cap \mathcal{T}|+1} = i_1$. The first term in the right hand side of (\ref{eq:pclaw}) is the well known expression for the average queueing delay for M/D/1 queues, weighed by $\rho_S$. The second and third terms in the right hand side of (\ref{eq:pclaw}) represent the average delay increase incurred by the travel time required for the collectors to move from one task to the other, i.e., the delay resulting from the switchover period. Further, for any coalition $S$ that must form in the system, the following condition must hold:
\begin{equation}\label{eq:polcond}
\rho_S < 1.
\end{equation}
This condition is a requirement for the stability of any polling system \cite{PS00,PS01,PS02} and, thus, must be satisfied for any coalition that will form in the proposed model. In the event where this condition is violated, the system is considered unstable and the delay is considered as infinite. In this regard, the analysis presented in the remainder of this paper will take into account this condition and its impact on the coalition formation process (as seen later, a coalition where $\rho_S \ge 1$ will never form).\vspace{-0.45cm}

\subsection{Utility Function}\vspace{-0.1cm}
In the proposed game, for every coalition $S \subseteq \mathcal{N}$, the agents must determine the order in which the tasks in $S$ are visited, i.e., the path $\{i_1,\ldots,i_{|S\cap \mathcal{T}|}\}$ which is an ordering over the set of tasks in $S$ given by $S \cap \mathcal{T}$. Naturally, the agents must select the path that minimizes the total switchover time for one round of data collection. This can be mapped to the following well-known problem:
  \begin{property}\label{prop:sales}
  The problem of finding the path that minimizes the total switchover time for one round of data collection within a coalition $S \subseteq \mathcal{N}$ is mapped into the \emph{traveling salesman} problem \cite{TSP00}, where a salesman, i.e., the agents $S \cap \mathcal{M}$, is required to minimize the time of visiting a series of cities, i.e., the tasks $S\cap \mathcal{T}$.
  \end{property}

It is widely known that the solution for the traveling salesman problem is NP-complete~\cite{TSP00}, and, hence, there has been numerous heuristic algorithms for finding an acceptable near-optimal solution. One of the simplest of such algorithms is the nearest neighbor algorithm (also known as the greedy algorithm) \cite{TSP00}. In this algorithm, starting from a given city the salesman chooses the closet city as his next visit. Using the nearest neighbor algorithm, the ordering of the cities which minimizes the overall route is selected. The nearest neighbor algorithm is sub-optimal but it can quickly find a near-optimal solution with a small computational complexity (linear in the number of cities) \cite{TSP00} which makes it suitable for problems such as the proposed task allocation problem. Therefore, in the proposed model, for every coalition $S$, the agents can easily work out the nearest neighbor route for the tasks, and operate according to it.

Having modeled each coalition as a polling system, the pseudo-conservation law in (\ref{eq:pclaw}) allows to evaluate the cost, in terms of average waiting time (or delay), from forming a particular coalition. However, for every coalition, there is a benefit, in terms of the average effective throughput that the coalition is able to achieve. The average effective throughput for a coalition $S$ is given by
\begin{equation}\label{eq:effthr}
L_{S} = \sum_{i \in S\cap \mathcal{T}} \lambda_i  \cdot \textrm{Pr}_{i,\textrm{CR}},
\end{equation}
with $\textrm{Pr}_{i,\textrm{CR}}$ given by (\ref{eq:suc}). By closely inspecting (\ref{eq:srate}), one can see that adding more collectors improves the transmission link capacity, and, thus, reduces the service time that a certain task perceives. Based on this property and by using (\ref{eq:pclaw}) one can easily see that, adding more collectors, i.e., improving the service time, reduces the overall delay in (\ref{eq:pclaw})  \cite{PS00,PS01,PS02}. Further, adding more relays would reduce the distance over which transmission is occurring, thus, improving the probability of successful transmission as per (\ref{eq:suc}) \cite{AL00,DM00}. In consequence, using (\ref{eq:effthr}), one can see that this improvement in the probability of successful transmission is translated into an improvement in the effective throughput. Hence, each agent role (collector or relay) possesses its own benefit for a coalition.

A suitable criterion for characterizing the utility in networks that exhibit a tradeoff between the throughput and the delay is the concept of system \emph{power} which is defined as the ratio of some power of the throughput and the delay (or a power of the delay) \cite{KL00}. Hence, the concept of power is an attractive notion that allows to capture the fundamental tradeoff between throughput and delay in the proposed task allocation model. Power has been used thoroughly in the literature in applications that are sensitive to throughput as well as delay \cite{PW00,PW01,PW02}. Mainly, for the proposed game, the utility of every coalition $S$ is evaluated using a coalitional value function based on the power concept from \cite{PW02} as follows
\begin{equation}\label{eq:util}
v(S) =  \begin{cases} \delta \frac{L_{S}^{\beta}}{(\sum_{i \in S\cap \mathcal{T}} \rho_i \bar{W}_i)^{(1-\beta)}}, & \mbox{if } \rho_S < 1\ \textrm{and} \ |S| > 1,\\ 0, &\mbox{otherwise},\end{cases}
\end{equation}
where $\beta \in (0,1)$ is a throughput-delay tradeoff parameter. In (\ref{eq:util}), the term $\delta$ represents the price per unit power that the network offers to coalition $S$. Hence, $\delta$ represents a generic control parameter that allows the network operator to somehow monitor the behavior of the players. The use of such control parameters is prevalent in game theory \cite{CF00,CF01,HC00,HC01,HC02}. In certain scenarios, $\delta$ would represent physical monetary values paid by the operator to the different entities (agents and tasks). In such a case, on one hand, for the tasks, the operator simply would pay the tasks' owners for the amount of data (and its corresponding quality as per (\ref{eq:util})) each one of their tasks had generated. On the other hand, for the agents, the payment would, for example, represent either a reward for the behavior of the agents or the proportion of maintenance or servicing that each agent would receiver from its operator. In this sense, the utility function in (\ref{eq:util}) would, thus, represents the total revenue achieved by a coalition $S$, given the network power that coalition $S$ obtains. For coalitions that consist of a single agent or a single task, i.e., coalitions of size $1$, the utility assigned is $0$ due to the fact that such coalitions generate no benefit for their member (a single agent can collect no data unless it moves to task, while a single task cannot transmit any of its generated data without an agent collecting this data). Further, any coalition where condition (\ref{eq:polcond}) is not satisfied is also given a zero utility, since, in this case, the polling system that the coalition represents is unstable, and hence has an infinite delay.

Given the set of players $\mathcal{N}$, and the value function in (\ref{eq:util}), we define a coalitional game $(\mathcal{N},v)$ with transferable utility~(TU). The utility in (\ref{eq:util}) represents the amount of money or revenue received by a coalition, and, hence, this amount can be arbitrarily apportioned between the coalition members, which justifies the TU nature of the game. To divide this utility between the players, we adopt the equal fair allocation rule, where the payoff of any player $i \in S$, denoted by $x_i^S$ is
\begin{equation}\label{eq:eqfair}
x_i^S = \frac{v(S)}{|S|}.
\end{equation}

The payoff $x_i^S$ represents the amount of revenue that player $i \in S$ receives from the total revenue $v(S)$ that coalition $S$ generates. The main motivation behind adopting the equal fair allocation rule is in order to highlight the fact that the agents and the tasks value each others equally. As seen in (\ref{eq:util}), the presence of an agent in a coalition is crucial in order for the tasks to obtain any payoff, and, vice versa, the presence of a task in a coalition is required for the agent to be able to obtain any kind of utility. Nonetheless, the proposed model and algorithm can accommodate any other type of payoff allocation rule. Although in traditional coalitional games, the allocation rule may have a strong impact on the game's solution, for the proposed game, other allocation rules can be used with little impact on the analysis that is presented in the rest of the paper. This is due to the class of the proposed game which is quite different from traditional coalitional games. As seen from (\ref{eq:pclaw}) and (\ref{eq:util}), whenever the number of tasks in a coalition increases, the total delay increases, hence, reducing the utility from forming a coalition. Further, in a coalition where the number of tasks is large, the condition of stability for the polling system, as given by (\ref{eq:polcond}), can be violated due to heavy traffic incoming from a large number of tasks yielding a zero utility as per (\ref{eq:util}). Hence, forming coalitions between the tasks and the agents entails a cost that can limit the size of a coalition. In this regard, traditional solution concepts for TU games, such as the core \cite{Game_theory2}, may not be applicable. In fact, in order for the core to exist, a TU coalitional game must ensure that the grand coalition, i.e., the coalition of all players will form. However, as seen in Fig.~\ref{fig:fir} and corroborated by the utility in (\ref{eq:util}), in general, due to the cost for coalition formation, the grand coalition will not form. Instead, independent and disjoint coalitions appear in the network as a result of the task allocation process. In this regard, the proposed game is classified as a \emph{coalition formation game} \cite{CF00,CF01,HC00,HC01,HC02}, and the objective is to find an algorithm that allows to form the coalition structure, instead of finding only a solution concept, such as the core, which aims mainly at stabilizing a grand coalition of all players.\vspace{-0.2cm}

\section{Task Allocation as a Hedonic Coalition Formation Game}\label{sec:gform}\vspace{-0.1cm}
%In this section, we map the proposed task allocation problem to a  hedonic coalition formation game with an underlying transferable utility, and we propose a distributed algorithm for forming the coalitions using concepts from hedonic games.

\subsection{Hedonic Coalition Formation: Concepts and Model}\vspace{-0.1cm}
Coalition formation games have been a topic of high interest in game theory \cite{CF00,CF01,HC00,HC01,HC02}. Notably, in \cite{HC00,HC01,HC02}, a class of coalition formation games known as \emph{hedonic coalition formation games} is investigated. This class of games entails several interesting properties that can be applied, not only in economics such as in \cite{HC00,HC01,HC02}, but also in wireless networks as we will demonstrate in this paper. The two key requirements for classifying a coalitional game as a \emph{hedonic} game are \cite{HC00}:
\begin{condition} - (Hedonic Conditions) - A coalition formation game is classified as hedonic if
 \begin{enumerate}
 \item The payoff of any player depends \emph{solely} on the members of the coalition to which the player belongs.
 \item The coalitions form as a result of the \emph{preferences} of the players over their possible coalitions' set.
     \end{enumerate}
     \end{condition}

These two conditions characterize the framework of hedonic games. Mainly, the term \emph{hedonic} pertains to the first condition above, whereby the payoff of any player $i$, in a hedonic game, must depend only on the identity of the players in the coalition to which player $i$ belongs, with no dependence on the other players. For the second condition, in the remainder of this section, we will formally define how the preferences of the players over the coalitions can be used for the formation process. To use hedonic games in the proposed model, we first introduce some definitions, taken from \cite{HC00}.
\begin{definition} A \emph{coalition structure} or a \emph{coalition partition} is defined as the set $\Pi = \{S_1,\ldots,S_l\}$ which partitions the players set $\mathcal{N}$, i.e.,  $ \forall\ k\ ,S_k \subseteq \mathcal{N}$ are disjoint coalitions such that $\cup_{k=1}^{l}S_k = \mathcal{N}$ (an example partition $\Pi$ is shown in Fig.~\ref{fig:fir}).
 \end{definition}
 \begin{definition}
 Given a partition $\Pi$ of $\mathcal{N}$, for every player $i\in \mathcal{N}$ we denote by $S_{\Pi}(i)$, the coalition to which player $i$ belongs, i.e., coalition $S_{\Pi}(i)=S_k \in \Pi$, such that $i \in S_k$.
 \end{definition}

In a hedonic game, each player must build preferences over its own set of possible coalitions. Hence, each player must be able to compare the coalitions and order them based on which coalition  prefers being a member of. To evaluate these players' preferences over the coalitions, we define the concept of a preference relation or order as follows \cite{HC00}:
\begin{definition}
For any player $i\in \mathcal{N}$, a \emph{preference relation} or \emph{order} $\succeq_i$ is defined as a complete, reflexive, and transitive binary relation over the set of all coalitions that player $i$ can possibly form, i.e., the set $\{S_k \subseteq \mathcal{N} : i \in S_k\}$.
\end{definition}

Thus, for a player $i \in \mathcal{N}$, given two coalitions $S_1 \subseteq \mathcal{N}$ and, $S_2 \subseteq \mathcal{N}$ such that $i \in S_1$ and $i \in S_2$,  $S_1\succeq_i S_2$ indicates that player $i$ prefers to be part of coalition $S_1$, over being part of coalition $S_2$, or at least, $i$ prefers both coalitions equally. Further, using the asymmetric counterpart of $\succeq_i$, denoted by $\succ_i$, then $S_1 \succ_i S_2$, indicates that player $i$ \emph{strictly} prefers being a member of $S_1$ over being a member of $S_2$. For every application, an adequate preference relation $\succeq_i$ can be defined to allow the players to quantify their preferences. The preference relation can be a function of many parameters, such as the payoffs that the players receive from each coalition, the weight each player gives to other players, and so on. Given the set of players $\mathcal{N}$, and a preference relation $\succeq_i$ for every player $i \in \mathcal{N}$, a hedonic coalition formation game is formally defined as follows \cite{HC00}:
\begin{definition}
A hedonic coalition formation game is a coalitional game that satisfies the two hedonic conditions previously prescribed, and is defined by the pair $(\mathcal{N},\succ)$ where $\mathcal{N}$ is the set of players ($|\mathcal{N}|=N$), and $\succ$ is a \emph{profile of preferences}, i.e., preference relations, $(\succeq_1,\ldots,\succeq_N)$ defined for every player in $\mathcal{N}$.
\end{definition}

Having defined the main components of hedonic coalition formation games, we utilize this framework in order to provide a suitable solution for the task allocation problem proposed in Section~\ref{sec:sysmodel}. The proposed task allocation problem is modeled as a $(\mathcal{N},\succ)$ hedonic game, where $\mathcal{N}$ is the set of agents and tasks and $\succ$ is a profile of preferences that we will shortly define. For the proposed game model, given a network partition $\Pi$ of $\mathcal{N}$, the payoff of any player $i$, depends only on the identity of the members of the coalition to which $i$ belongs. In other words, the payoff of any player $i$ depends solely on the players in the coalition in which player $i$ belongs $S_{\Pi}(i)$ (easily seen through the formulation of Section~\ref{sec:gmodel}). Hence, our game verifies the first hedonic condition.

Moreover, to model the task allocation problem as a hedonic coalition formation game, the preference relations of the players must be clearly defined. In this regard, we define two types of preference relations, a first type suited for indicating the preferences of the agents, and a second type suited for the tasks. Subsequently, for evaluating the preferences of any agent $i \in \mathcal{M}$, we define the following operation (this preference relation is common for all agents, hence we denote it by $\succeq_i\ =\ \succeq_{\mathcal{M}},\ \forall i \in \mathcal{M}$)
\begin{equation}\label{eq:prefuav}
S_2 \succeq_{\mathcal{M}} S_1 \Leftrightarrow u_i(S_2) \ge u_i(S_1),
\end{equation}
where $S_1 \subseteq \mathcal{N}$ and $S_2 \subseteq \mathcal{N}$ are any two coalitions that contain agent $i$, i.e., $i \in S_1$ and $i \in S_2$ and $u_i:2^{\mathcal{N}}\rightarrow \mathbb{R}$ is a preference function defined for any agent $i$ as follows

\begin{equation}\label{eq:pref1}
u_i(S) = \begin{cases} \infty, & \mbox{if } S=S_\Pi(i)\ \& \  S\setminus\{i\} \subseteq \mathcal{T}, \\ 0, & \mbox{if } S \in h(i),\\ x_i^S. &\mbox{otherwise}, \end{cases}
\end{equation}
where  $\Pi$ is the \emph{current} coalition partition which is in place in the game, $x_i^S$ is the payoff received by player $i$ from any division of the value function among the players in coalition $S$ such as the equal fair division given in (\ref{eq:eqfair}), and $h(i)$ is the history set of player $i$. At any point in time, the history set $h(i)$ is a set that contains the coalitions that player $i$ was a part of in the past, prior to the formation of the current partition $\Pi$. Note that, by using the defined preference relation, the players can compare any two coalitions $S_1$ and $S_2$ independently of whether these two coalitions belong to $\Pi$ or not.

The main rationale behind the preference function $u_i$ in (\ref{eq:pref1}) is as follows. In this model, the agents, being entities owned by the operator, seek out to achieve two conflicting objectives: (i)- Service all tasks in the network for the benefit of the operator,  and (ii)- Maximize the quality of service, in terms of power as per (\ref{eq:util}), for extracting the data from the tasks. The preference function $u_i$ must be able to allow the agents to make coalition formation decisions that can capture this tradeoff between servicing all tasks (at the benefit of the operator) and achieving a good quality of service (at the benefit of both agents and operator). For this purpose, as per (\ref{eq:pref1}),  any agent $i$ that is the \emph{sole} agent servicing tasks in its current coalition $S=S_\Pi(i)$ such that $S_\Pi(i)\cap \mathcal{M} = \{i\}$, assigns an infinite preference value to $S$. Hence, in order to service all tasks, the agent always assigns a maximum preference to its current coalition, if this current coalition is composed of only tasks and does not contain other agents. This case of the preference function $u_i$ forbids the agent from leaving a group of tasks, already assigned to it, unattended by other agents. In this context, this condition pertains to the fist objective (objective (i) previously mentioned) of the agents and it implies that, whenever there is a risk of leaving tasks without service, the agent do not act selfishly, in contrast, they act in the benefit of the operator and remain with these tasks, independent of the payoff generated by these tasks. Such a decision  allows an agent to avoid making a decision that can incur a risk of ultimately having tasks with no service in the network, in which case, the network operator would lose revenue from these unattended tasks and it may, for example, decide to replace the agent that led to the presence of such a group of tasks with no service. Otherwise, the agents' preference relation $u_i$ would highlight the second objective of the agent, i.e., maximize its own payoff which maps into the revenue generated from the quality of service, i.e., the power as per (\ref{eq:util}). with which the tasks are being serviced. In this case, the preference relation is easily generated by the agents by comparing the value of the payoffs  they receive from the two coalitions $S_1$ and $S_2$. Further, we note that no agent has any incentive to revisit a coalition that it had left previously, and hence, the agents assign a preference value of $0$ for any coalition in their history (this can be seen as a basic learning process). In summary, taking into account the conflicting goals of the agents, between two coalitions $S_1$ and $S_2$, an agent $i$ prefers the coalition that gives the better payoff, given that the agent is not alone in its current coalition, and the coalition with a better payoff is not in the history of agent $i$.

To evaluate the preferences of any task $j \in \mathcal{T}$, we define the following operation (this preference relation is common for all tasks, hence, we denote it by $\succeq_j\ =\ \succeq_{\mathcal{T}},\ \forall j \in \mathcal{T}$)
 \begin{equation}\label{eq:preftask}
S_2 \succeq_{\mathcal{T}} S_1 \Leftrightarrow w_j(S_2) \ge w_j(S_1),
\end{equation}
with the tasks' preference function $w_j$ defined as follows.
\begin{equation}\label{eq:pref2}
w_j(S) = \begin{cases} 0, & \mbox{if } S \in h(j),\\ x_j^S, &
\mbox{otherwise}. \end{cases}
\end{equation}
The preferences of the tasks are easily captured using the function $w_j$. The preference function of the tasks is different from that of the agents since the tasks are, in general, independent entities that act solely in their own interest. Thus, based on (\ref{eq:pref2}), each task prefers the coalition that provides the larger payoff $x_j^S$ unless this coalition was already visited previously and left. In that case, the preference function of the tasks assigns a preference value of $0$ for any coalition that the task had already visited in the past (and left to join another coalition). Using this preference relation, every task can evaluate its preferences over the possible coalitions that the task can form. Consequently, the proposed task allocation model verifies both hedonic conditions, and, hence, we have:
\begin{property}
The proposed task allocation problem among the agents is modeled as a $(\mathcal{N},\succ)$ hedonic coalition formation game, with the preference relations given by (\ref{eq:prefuav}) and (\ref{eq:preftask}).
\end{property}

Note that the preference relations in  (\ref{eq:prefuav}) and (\ref{eq:preftask}) are also dependent on the underlying TU coalitional game described in Section~\ref{sec:gmodel}. Having formulated the problem as a hedonic game, the final task is to provide a distributed algorithm, based on the defined preferences, for forming the coalitions. However, prior to deriving the algorithm for coalition formation, we highlight the following result:
\begin{proposition}\label{prop:property2}
For the proposed hedonic coalition formation model for task allocation, assuming that all collector-agents have an equal link transmission capacity $\mu_i = \mu, \ \forall i \in \mathcal{M}$, any coalition $S \subseteq \mathcal{N}$ with $| S \cap \mathcal{M}|$ agents, must have at least $|G_S|_\textrm{min}$ collector-agents ($G_S \subseteq S \cap \mathcal{M}$) as follows
\begin{equation}
|G_S| > |G_S|_\textrm{min} = \frac{\sum_{i\in S\cap \mathcal{T}} \lambda_i }{ \mu}.
\end{equation}
Further, when all the tasks in $S$ belong to the same class, we have
\begin{equation}
|G_S|_\textrm{min} = \frac{ |S \cap \mathcal{T}| \cdot \lambda}{ \mu},
\end{equation}
 which constitutes an upper bound on the number of collector-agents as a function of the number of tasks $|S \cap \mathcal{T}|$ for a given coalition $S$.
\end{proposition}
\begin{proof}
 As per the defined preference relations in (\ref{eq:pref1}) and (\ref{eq:pref2}), any coalition that will form in the proposed model must be stable since no agent or task has an incentive to join an unstable coalition, hence, we have, for every coalition $S \subseteq \mathcal{N}$ having $|G_S|$ collectors with $G_S \subseteq S\cap \mathcal{M}$, we have from (\ref{eq:polcond})  $\rho_S < 1$, and thus
 \begin{equation*}
 \sum_{i \in S \cap \mathcal{T}} \frac{\lambda_i}{\mu^{G_S}} < 1,
 \end{equation*}
 which, given the assumption that $\mu_i = \mu, \forall i \in \mathcal{M}$ yields
  \begin{equation*}
 \frac{1}{|G_S|\cdot\mu}\cdot \sum_{i \in S \cap \mathcal{T}} \lambda_i < 1,
 \end{equation*}
 which yields
%  \begin{equation}\label{eq:pro}
$|G_S| > |G_S|_\textrm{min} = \frac{\sum_{i\in S\cap \mathcal{T}} \lambda_i}{ \mu}$.
 %\end{equation}
 Further, if we assume that all the tasks belong to the same class, hence, $\lambda_i = \lambda, \ \forall i \in S \cap \mathcal{T}$, we immediately get
 \begin{equation}
|G_S|_\textrm{min} = \frac{|S \cap \mathcal{T}| \cdot \lambda}{ \mu }.
\end{equation}
\end{proof}
Consequently, for any proposed coalition formation algorithm, the bounds on the number of collector-agents in any coalition $S$ as given by Proposition~\ref{prop:property2} will always be satisfied.

\subsection{Hedonic Coalition Formation: Algorithm}\vspace{-0.2cm}
In the previous subsection, we modeled the task allocation problem as a hedonic coalition formation game and, thus, the remaining objective is to devise an algorithm for forming the coalitions. While literature that studies the characteristics of existing partitions in hedonic games, such as in \cite{HC00,HC01,HC02}, is abundant, the problem of forming the coalitions both in the hedonic and non-hedonic setting is a challenging problem \cite{CF00}. In this paper, we introduce an algorithm for coalition formation that allows the players to make \emph{selfish} decisions as to which coalitions they decide to join at any point in time. In this regard, for forming coalitions between the tasks and the agents, we propose the following rule for coalition formation:
\begin{definition}\label{def:switch}
\textbf{Switch Rule -} Given a partition $\Pi=\{S_1,\ldots,S_l\}$ of the set of players (agents and tasks) $\mathcal{N}$, a Player $i$ decides to leave its current coalition $S_{\Pi}(i)=S_m,\ $ for some $m \in \{1,\ldots,l\}$ and join another coalition $S_k \in \Pi \cup \{\emptyset\},\ S_k \neq S_{\Pi}(i)$, if and only if $S_k \cup \{i\} \succ_i S_{\Pi}(i)$. Hence, $\{S_m,S_k\} \rightarrow \{S_m\setminus\{i\},S_k\cup\{i\}\}$.
\end{definition}

Through a single switch rule made by any player $i$, any current partition $\Pi$ of $\mathcal{N}$ is transformed into $\Pi^{\prime} = (\Pi \setminus \{S_m,S_k\}) \cup \{S_m\setminus\{i\},S_k\cup\{i\}\}$. In simple terms, for every partition $\Pi$, the switch rule provides a mechanism through which any task or agent, can leave its current coalition $S_{\Pi}(i)$, and join another coalition $S_k \in \Pi$, given that the new coalition $S_k \cup \{i\}$ is strictly preferred over $S_{\Pi}(i)$ through any preference relation that $i$ is using (in particular using the preference relations defined in (\ref{eq:prefuav}) and (\ref{eq:preftask})). Independent of the preference relations selected, the switch rule can be seen as a \emph{selfish} decision made by a player, to move from its current coalition to a new coalition, regardless of the effect of this move on the other players. Furthermore, we consider that, whenever a player decides to switch from one coalition to another, the player updates its history set $h(i)$. Hence, given a partition $\Pi$, whenever a player $i$ decides to leave coalition $S_m \in \Pi$ to join a different coalition, coalition $S_m$ is automatically stored by player $i$ in its history set $h(i)$.

Consequently, we propose a coalition formation algorithm composed of three main phases: Task discovery, hedonic coalition formation, and data collection. In the first phase, the central command receives information about the existence of tasks that require servicing and informs the agents of the locations and characteristics of the tasks (e.g., the arrival rates). Hence, the agents start by having full knowledge of the initial partition $\Pi_{\textrm{initial}}$. Once the agents are aware of the tasks, they broadcast their own presence to the tasks. Consequently, the players can interact with each other, for performing coalition formation. Hence, the second phase of the algorithm is the hedonic coalition formation phase. In this phase, all the players (tasks and agents) investigate the possibility of performing a switch operation. For identifying potential switch operations, given complete knowledge about the network (which can be gathered in different methods as will be discussed in Subsection~\ref{sec:impl})), each agent investigates its top preference, and decides to perform a switch operation, if possible through (\ref{eq:prefuav}).  As one can easily see in (\ref{eq:util}), for the proposed model, no coalition composed of tasks-only would ever form since such a coalition would always generate a $0$ utility. Therefore, the tasks are only interested in switching to coalitions that contain at least a single agent.
From a tasks' perspective, for determining its preferred switch operation, each task needs only to negotiate with existing agents in order to enquire on the amount of utility it can obtain by joining with this agent. By doing so, each task can determine the switch operation it is interested in making at a given time. We consider that, the players make sequential switch decisions (the order of switch operations is referred to as the \emph{order of play} hereinafter). For any agent, a switch operation is easily performed as the agent can leave its current coalition and join the new coalition, if (\ref{eq:prefuav}) is satisfied. For the tasks, any task that finds out a possibility to switch, can autonomously request, over a control channel with the concerned agent, to be added to the coalition of interest (which would always contain at least one agent with whom the task previously negotiated). The convergence of the proposed algorithm during the hedonic coalition formation phase is guaranteed as follows:
\begin{theorem}\label{th:one}
Starting from any initial network partition $\Pi_{\textrm{initial}}$, the proposed hedonic coalition formation phase of the proposed algorithm always converges to a final network partition $\Pi_f$ composed of a number of disjoint coalitions.
\end{theorem}
\begin{proof}
For the purpose of this proof, we denote $\Pi_{n_{k}}^{k}$ as the partition formed during the time $k$ when player $i\in \mathcal{N}$ decides to act after $n_{k}$ switch operations have previously occurred (the index $n_{k}$ denotes the number of switch operations performed by one or more players up to time $k$).  Given any initial starting partition $\Pi_{\textrm{initial}}=\Pi^{1}_{0}$, the hedonic coalition formation phase of the proposed algorithm consists of a sequence of switch operations. As per Definition~\ref{def:switch}, every switch operation transforms the current partition $\Pi$ into another partition $\Pi^\prime$, hence, hedonic coalition formation consists of a sequence of switch rules, yielding, e.g., the following transformations
\begin{equation}\label{eq:trans}
\Pi^{1}_{0}=\Pi^{2}_{0}\rightarrow \Pi^{3}_{1} \rightarrow \ldots \rightarrow \Pi^{L}_{n_{L}} \ldots\rightarrow \ldots \rightarrow \Pi^{T}_{n_T},
\end{equation}
where the operator $\rightarrow$ indicate the occurrence of a switch operation. In other words, $\Pi_{n_{k}}^{k} \rightarrow \Pi_{n_{k+1}}^{k+1}$, implies that during turn $k$, a certain player $i$ made a single switch operation which yielded a new partition $\Pi_{n_{k+1}}^{k+1}$ at the turn $k+1$. By inspecting the preference relations defined in (\ref{eq:prefuav}) and (\ref{eq:preftask}), it can be seen that every single switch operation leads to a partition that has not yet been visited (new partition). Hence, for any two partitions $\Pi_{n_{k}}^{k}$ and $\Pi_{n_{l}}^{l}$ in the transformations of (\ref{eq:trans}), such that $n_{k} \neq n_{l}$, i.e., $\Pi_{n_{l}}^{l}$ is a result of the transformation of $\Pi_{n_{k}}^{k}$ (or vice versa) after a number of switch operations $|n_l - n_k|$, we have that $\Pi_{n_{k}}^{k} \neq \Pi_{n_{l}}^{l}$ for any two turns $k$ and $l$.

Given this property and the well known fact that the number of partitions of a set is \emph{finite} and given by the Bell number \cite{CF00}, the number of transformations in (\ref{eq:trans}) is finite, and, hence, the sequence in (\ref{eq:trans}) will always terminate and converge to a final partition $\Pi_f = \Pi^{T}_{n_{T}}$ which completes the proof.
\end{proof}
 \begin{table}[!t]
\scriptsize
  \centering
  \caption{%\mycaption{%\vspace*{-1em}
    \vspace*{-0.2em}The proposed hedonic coalition formation algorithm for task allocation in wireless networks.}\vspace*{-1.5em}
    \begin{tabular}{p{8cm}}
      \hline
      \textbf{Initial State} \vspace*{.5em} \\
      \hspace*{1em} The network is partitioned by $\Pi_{\textrm{initial}}=\{S_1,\ldots,S_k\}$. At the beginning of all time $\Pi_{\textrm{initial}}$ = $\mathcal{N}$ = $\mathcal{M} \cup \mathcal{T}$ with no tasks being serviced. \vspace*{.1em}\\
\textbf{Three Phases for the Proposed Hedonic Coalition Formation Algorithm} \vspace*{.1em}\\
\hspace*{1em}\emph{Phase I - Task Discovery:}   \vspace*{.1em}\\
\hspace*{1.5em}a) The command center is informed by one or multiple owners about the\vspace*{.1em}\\
\hspace*{1.5em}existence and characteristics of new tasks.\vspace*{.1em}\\
\hspace*{1.5em}b) The central command center conveys the information on the initial network\vspace*{.1em}\\
\hspace*{1.5em}partition $\Pi_{\textrm{initial}}$ using the methods of Subsection~\ref{sec:impl}\vspace*{.1em}\\
\hspace*{1em}\emph{Phase II - Hedonic Coalition Formation:}   \vspace*{.1em}\\
\hspace*{3em}\textbf{repeat}\vspace*{.2em}\\
\hspace*{3em}For every player $i \in \mathcal{N}$, given a current partition $\Pi_{\textrm{current}}$\vspace*{.2em}.\vspace*{.2em}\\
\hspace*{4em}a) Player $i$ investigates possible switch using the preferences\vspace*{.2em}\\
\hspace*{4em}given, respectively, by (\ref{eq:prefuav}) and (\ref{eq:preftask}) for the agents and tasks.\vspace*{.2em}\\
\hspace*{4em}b) Player $i$ performs the switch operation that maximizes its payoff:\vspace*{.2em}\\
\hspace*{5em}b.1) Player $i$ updates its history $h(i)$ by adding $S_{\Pi_{\textrm{current}}}(i)$.\vspace*{.2em}\\
\hspace*{5em}b.2) Player $i$ leaves its current coalition $S_{\Pi_{\textrm{current}}}(i)$.\vspace*{.2em}\\
\hspace*{5em}b.3) Player $i$ joins the new coalition that maximizes its payoff.\vspace*{.2em}\\
\hspace*{3em}\textbf{until} convergence to a final Nash-stable partition $\Pi_{f}$.\vspace*{.2em}\\
\hspace*{1em}\emph{Phase III - Data Collection}   \vspace*{.1em}\\
\hspace*{2em}a) The network is partitioned using $\Pi_{\textrm{final}}$.\vspace*{.1em}\\
\hspace*{2em}b) The agents in each coalition $S_k \in \Pi_{\textrm{final}}$ \emph{continuously} perform the\vspace*{.1em}\\
\hspace*{2em}following operations, i.e., act as a polling system with exhaustive strategy\vspace*{.1em}\\
\hspace*{2em}and switchover times:\vspace*{.1em}\\
\hspace*{3em}b.1) Visit a first task in their respective coalitions.\vspace*{.1em}\\
\hspace*{3em}b.2) The collector-agents collect the data from the task being visited.\vspace*{.1em}\\
\hspace*{3em}b.3) The collector-agents transmit the data using wireless links to the\vspace*{.1em}\\
\hspace*{3em}central receiver either directly or using other relay-agents.\vspace*{.1em}\\
\hspace*{3em}b.4) Once the queue of the current is empty, visit the next task.\vspace*{.1em}\\
\hspace*{2em}The order in which the tasks are visited is determined by the nearest \vspace*{.1em}\\
\hspace*{2em}neighbor solution to the traveling salesman problem as in Property~\ref{prop:sales}.\vspace*{.1em}\\
\hspace*{2em}This third phase is continuously repeated and performed by all the agents\vspace*{.1em}\\
\hspace*{2em}in $\Pi_{\textrm{final}}$ for a fixed period of time $\Psi$ (for static environments $\Psi=\infty$).\vspace*{.1em}\\
\hspace*{1em}\emph{Adaptation to environmental changes (periodic process)}   \vspace*{.1em}\\
\hspace*{2em}a) In the presence of environmental changes, such as the deployment of\vspace*{.1em}\\
\hspace*{2em}new tasks, the removal of existing tasks, or periodic low mobility of\vspace*{.1em}\\
\hspace*{2em}the tasks, the third phase of the algorithm is performed continuously\vspace*{.1em}\\
\hspace*{2em}only for a \emph{fixed} period of time $\Psi$.\vspace*{.1em}\\
\hspace*{2em}b) After $\Psi$ elapses, the first two phases are repeated to allow the players\vspace*{.1em}\\
\hspace*{2em}to self-organize and adapt the network to these environmental changes.\vspace*{.1em}\\
\hspace*{2em}c) This process is repeated periodically for networks where environmental\vspace*{.1em}\\
\hspace*{2em}changes may occur.\vspace*{.1em}\\
   \hline
    \end{tabular}\label{tab:alg}\vspace{-0.7cm}
\end{table}

The stability of the final partition $\Pi_f$ resulting from the convergence of the proposed algorithm can be studied using the following stability concept from hedonic games  \cite{HC00}:
\begin{definition}
A partition $\Pi = \{S_1,\ldots,S_l\}$ is  \emph{Nash-stable} if $\forall i \in \mathcal{N},\  S_{\Pi}(i) \succeq_i S_k \cup \{i\}$ for all $S_k \in \Pi \cup \{\emptyset\}$ (for agents $\succeq_i = \succeq_{\mathcal{M}}, \forall i \in \mathcal{N} \cap \mathcal{M}$ and for tasks $\succeq_i = \succeq_{\mathcal{T}}, \forall i \in \mathcal{N} \cap \mathcal{T}$).
\end{definition}
In other words, a coalition partition $\Pi$ is Nash-stable, if no player has an incentive to move from its current coalition to another coalition in $\Pi$ or to deviate and act alone. %Note that, the above definition requires Nash stability using the strict preference relation $\succ_i$ unlike the looser definition in \cite{HC00} which is based on the preference relation $\succeq_i$.
Furthermore, a Nash-stable partition $\Pi$ implies that there does not exist any coalition $S_k \in \mathcal{N}$ such that a player $i$ strictly prefers to be part of $S_k$ over being part of its current coalitions, while all players of $S_k$ do not get hurt by forming $S_k \cup \{i\}$. This is the concept of individual stability, which is formally defined as follows \cite{HC00}:
\begin{definition}
A partition $\Pi = \{S_1,\ldots,S_l\}$ is \emph{individually stable} if there do not exist $i \in \mathcal{N},$ and a coalition $S_k \in \Pi \cup \{\emptyset\}$ such that $  S_k \cup \{i\} \succ_i S_{\Pi}(i)$ and $  S_k \cup \{i\} \succeq_j S_k$ for all $j \in S_k$ (for agents $\succeq_i = \succeq_{\mathcal{M}}, \forall i \in \mathcal{N} \cap \mathcal{M}$ and for tasks $\succeq_i = \succeq_{\mathcal{T}}, \forall i \in \mathcal{N} \cap \mathcal{T}$ for tasks).
\end{definition}
As already noted, a Nash-stable partition is individually stable \cite{HC00}. For the hedonic coalition formation phase of the proposed algorithm, we have the following:
\begin{proposition}\label{prop:one}
Any partition $\Pi_f$ resulting from the hedonic coalition formation phase of the proposed algorithm is Nash-stable, and, hence, individually stable.
\end{proposition}
\begin{proof}
For any partition $\Pi$, no player (agent or task) $i \in \mathcal{N}$ has an incentive to leave its current coalition, and act alone as per the utility function in (\ref{eq:util}). Assume that the partition $\Pi_f$ resulting from the proposed algorithm is not Nash-stable. Hence, there exists a player $i \in \mathcal{N}$, and a coalition $S_k \in \Pi_f$ such that $S_k \cup \{i\} \succ_i S_{\Pi_f}(i)$, hence, player $i$ can perform a \emph{switch} operation which contradicts with the fact that $\Pi_f$ is the result of the convergence of the proposed algorithm (Theorem~\ref{th:one}). Thus, any partition $\Pi_f$ resulting from the hedonic coalition formation phase of the proposed algorithm is Nash-stable, and, hence, by \cite{HC00}, this resulting partition is also individually stable.
\end{proof}
Following the formation of the coalitions and the convergence of the hedonic coalition formation phase to a Nash-stable partition, the last phase of the algorithm entails the actual data collection by the agents. In this phase, the agents move from one task to the other, in their respective coalitions, collecting the data and transmitting it to the central receiver, similar to a polling system, as explained in Sections~\ref{sec:sysmodel} and \ref{sec:gmodel}. A summary of the proposed algorithm is shown in Table~\ref{tab:alg}.

The proposed algorithm, as highlighted in Table~\ref{tab:alg}, can adapt the network topology to environmental changes such as the deployment of new tasks, the removal of a number of existing tasks, or a periodic low mobility of the tasks (in the case where the tasks represent mobile sensor devices for example). For this purpose, the first two phases of the algorithm shown in Table~\ref{tab:alg} are repeated periodically over time, to adapt to any changes that occurred in the environment. With regards to mobility, we only consider the cases where the tasks are mobile for a fixed period of time with a
velocity that is smaller than that of the agents $\eta$. In the presence of such a mobile environment, the central command center, through Phase~I of the algorithm in Table~\ref{tab:alg} informs the agents of the new tasks locations (periodically) and, thus, during Phase~II of the proposed algorithm, both agents and tasks can react to the environment changes, and modify the existing topology. As per Theorem~\ref{th:one} and Proposition~\ref{prop:one}, regardless of the starting position, the players will always self-organize into a Nash-stable partition, even after mobility, the deployment of new tasks or the removal of existing tasks. In summary, in a changing environment, the first two phases of the algorithm in Table~\ref{tab:alg} are repeated periodically, after a certain  fixed period of time $\Psi$ has elapsed during which the players were involved in Phase~III and the actual data collection and transmission occurred. Finally, whenever a changing environment is considered, the players are also allowed to periodically clear their history, so as to allow them to explore all the new possibilities that the changes in the environment may have yielded.\vspace{-0.3cm}

\subsection{Distributed Implementation Possibilities}\label{sec:impl}\vspace{-0.2cm}
 For implementation, as shown in Fig.~\ref{fig:fir}, we clearly distinguish between two inherently different entities: The command center, which is the intelligence that has some control over the agents and the central receiver which is a node in the network that is connected to the command center and which would receive the data transmitted by the agents (this distinction can be, for example, analogous to the distinction between a radio network controller and a base station in cellular networks). In practice, the central command can be, for example, a node that owns a number of agents and controls a large area which is divided into smaller areas with each area represented by the illustration of Fig.~\ref{fig:fir}. Hence, each such small area is a region having its own central receiver and where a subset of agents needs to operate and perform coalition formation using our model. In other scenarios, the command center can also be a satellite system that controls groups of agents with each group deployed in a different area (notably when the agents are UAVs for example). In contrast, the central receiver is simply a wireless node that receives the data from the agents and, subsequently, the command center can obtain this data from all receivers in its controlled area (e.g., through a backbone)\footnote{Our model can accommodate the case in which the command center and the central receiver coincide, e.g., in a small single-area network.}.

For performing coalition formation, the agents and tasks are required to know different types of information. In order to perform a switch operation, each agent is required to obtain data on the location of the tasks (hence, consequently deducing the hop distance $D_{ij}$ between any two tasks $i$ and $j$) as well as on the arrival rates $\lambda_i, i \in \mathcal{T}$ of these tasks. As a first step, whenever a tasks' owner (e.g., a service provider or a third party) requires that its tasks be serviced, it will give the details and characteristics of these tasks to the network operator (through service-level agreements for example) which would enter these details into the command center. Subsequently, the command center can insert this information into appropriate databases that the agents can access through, for example, an Internet connection. Such a transfer of information through active databases has been recently utilized in many communication architectures, for example, in cognitive radio network for primary user information distribution \cite{SD00}, or in UAVs operation \cite{DA01}. In cases where the command center controls only a single set of agents and a single area, this information can be, instead, broadcast directly to the agents. Further, the agents are also required to know the capabilities of each others, notably, the link transmission capacity $\mu_i, \forall i \in \mathcal{M}$ and the velocity (which can be used to deduce the switchover times). As the agents are all owned by a single operator, this information can be easily fed to the agents at the beginning of all time prior to their deployment, and, thus, does not require any additional communication during coalition formation.

From the tasks perspective, the amount of information that needs to be known is much less than that of the agents, notably since the tasks are, in general, resource-limited entities. For instance, as mentioned in the previous section, for performing coalition formation, the tasks do not need to know about the existence or the characteristics of each others. The main information that needs to be known by the tasks is the actual presence of agents. The agents can initially announce/broadcast their presence to the tasks as soon as they enter into the network. Subsequently, the tasks need only to be able to enquire, over a control channel, about the potential utility they would receive from joining the coalition of a particular agent (which can contain other tasks or agents but this is transparent from the perspective of the tasks). The main reason for this is that the tasks have no benefit in forming coalitions that have no agents since such coalitions generate $0$ utility for the tasks. Hence, from the point of view of the tasks, they would see every agent as a black box which can provide a certain payoff (communicated over a control channel during negotiation phase), and, based on this, they decide to join the coalition one or another agent (if multiple agents are in the same coalition then they would offer the same benefit from the tasks perspective). Note that, for coalitions that contain multiple agents, the task needs only to ask \emph{one} agent about their potential utility. In fact, this agent can append, along with the information on the utility, a signal to the tasks about other agents that belong to the same coalition. By doing so, the tasks would no longer need to assess whether to join a coalition by enquiring from other agents that belong to the same coalition, i.e., having redundant information. Hence, by sending this additional information, the agent will enable the tasks to avoid doing multiple processing for the same enquiry.

Given the information that needs to be known by each player, the proposed algorithm can be implemented in a distributed way since the switch operations can be performed by the tasks or the agents independently of any centralized entity. In this regard, given a partition $\Pi$, in order to determine its preferred switch operation, an agent would assess the payoff it would obtain by joining with any coalition in $\Pi$, except for singleton coalitions composed of agents only. For the tasks, given $\Pi$, each task negotiates with only the agents (and the coalition to which they belong) in the network in order to evaluate its payoff and decide on a switch operation. By adopting a distributed implementation, one would reduce the overhead and computational load on the command center, notably when this command center is controlling numerous areas with different groups of agents (each such area is represented by the model of Fig.~\ref{fig:fir}). Further, the distributed approach allows to decentralize the intelligence, and, thus, reduces the detrimental effects  on the network and the tasks' owners that can be caused by failures or malicious behavior at the command center level. It is also important to note that the distributed approach complies better with the nature of both the agents and the tasks. On one hand, the agents are inherently autonomous nodes (partially controlled by the command center) that need to operate on their own and, thus, make distributed decisions \cite{MF00,MF03,CR02,AL00,AL01,AL02}. On the other hand, the tasks are independent entities that belong to different owners. Consequently, the tasks are apt to make their own decisions regarding coalition formation and are, generally, unwilling to accept a coalitional structure imposed by an external entity such as the command center. Nonetheless, we note that a centralized approach can also be adopted for the proposed algorithm notably in small networks where, for example, the command center coincides with the central receiver and owns all the tasks.

Regarding complexity, the main complexity lies in the switch operation, the solution to the traveling salesman problem, i.e., determining the order in which the tasks are visited within a coalition in order to evaluate the utility function, and the assignment of agents as either collectors or relays. For instance, given a present coalitional structure $\Pi$ where each coalition in $\Pi$ has \emph{at least one task}, for every agent, the computational complexity of finding its next coalition, i.e., performing a switch operation, is easily seen to be $O(|\Pi|)$, and the worst case scenario is when all the tasks act alone, in that case $|\Pi| = T$. In contrast, for the tasks, the worst case complexity is $O(M)$ since, in order to make a switch operation, the tasks need only to negotiate with agents. With regards to the traveling salesman solution, the complexity of the used nearest neighbor solution is well known to be linear in the number of cities, i.e., tasks \cite{TSP00}, hence, for a coalition $S_k \in \Pi$, the complexity of finding the traveling salesman solution is simply $O(|S_k\cap \mathcal{T}|)$, where $S_k \cap \mathcal{T}$ is the set of tasks inside coalition $S_k$.  During coalition formation, i.e., Phase~II of the algorithm, whenever a coalition $S_k$ is formed, this coalition needs to compute its own traveling salesman problem, which has a linear complexity $O(|S_k\cap \mathcal{T}|)$, as already mentioned. Certainly, the overall number of traveling salesman problems that should be solved also depends on the number of \emph{new} coalitions  that were potentially evaluated for coalition formation prior to convergence to the Nash-stable partition. Hence, the number of traveling salesman solutions that need to be computed is proportional to the number of coalitions (and the identity and number of the tasks within) that negotiated a potential coalition formation \emph{prior} to convergence. This certainly depends on the number of iterations till convergence and the number of switch operations that occurred. Nonetheless, for static environments, \emph{after coalition formation ends}, i.e., in Phase~III of the algorithm in Table~\ref{tab:alg}, the traveling salesman solution needs to be computed only once \emph{for each coalition} and, afterwards, the network can operate indefinitely (if the environment is static) without any need for the coalitions to re-compute the traveling salesman solution.

Additionally, for determining whether a agent acts as a collector or relay within any coalition, we consider that the players would compute this configuration by inspecting all combinations and selecting the one that maximizes the utility in (\ref{eq:util}). This computation is done during coalition formation for evaluating the potential utility, and, upon convergence, is maintained during network operation. As the number of agents in a single coalition is generally small, this computation is straightforward, and has reasonable complexity. Finally, in dynamic environments, as the algorithm is repeated periodically and since we consider only periodic low mobility, the complexity of hedonic coalition formation is comparable to the static environment case, but with more runs of the algorithm.\vspace{-0.2cm}

\section{Simulation Results and Analysis}\label{sec:sim}\vspace{-0.2cm}
For simulations, the following network is set up: A central receiver is
placed at the origin of a $4$~km $\times 4$~km square area with the tasks appearing in the area around it. The path loss parameters are set to $\alpha=3$ and $\kappa = 1$, the target SNR is set to $\nu_0=10$~dB, the pricing factor is set to $\delta=1$, and the noise variance $\sigma^2=-120$~dBm. All packets are considered of size $256$~bits which is a typical IP packet size. The agents are considered as having a constant velocity of $\eta=60$~km/h, a transmit power of $\tilde{P}=100$~mW, and a transmission link capacity of $\mu = 768$~kbps (assumed the same for all agents). Further, we consider two classes of tasks in the network. A first class that can be mapped to voice services having an arrival rate of $32$~kbps, and a second class that can be mapped to video services, such as the widely known Quarter Common Intermediate Format~(QCIF)  \cite{ST00}, having an arrival rate $128$~kbps. Tasks belonging to each class are generated with equal probability in the simulations. Unless stated otherwise, the throughput-delay tradeoff parameter $\beta$ is set to $0.7$, to indicate services that are reasonably delay tolerant. %All the statistical results shown in this section are averaged over the random positions of the tasks through a large number of runs.
\begin{figure}[!t]
\begin{center}
\includegraphics[width=80mm]{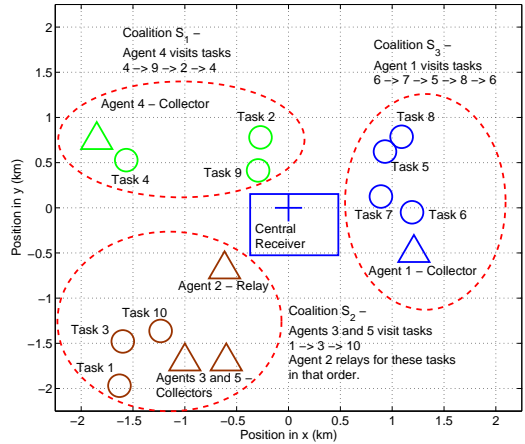}
\end{center}\vspace{-0.7cm}
\caption {A snapshot of a final coalition structure resulting from the proposed algorithm for a network of $M=5$~agents and $T=10$~tasks. In every coalition, the agents (collectors and relays) visit the tasks continuously in the shown order.} \label{fig:snapshot}\vspace{-0.7cm}
\end{figure}

In Fig.~\ref{fig:snapshot}, we show a snapshot of the final network partition reached through the proposed hedonic coalition formation algorithm for a network consisting of $M=5$~agents and $T=10$~arbitrarily located tasks. In this figure, Tasks $1$, $3$, and $8$ belong to the QCIF video class with an arrival rate of $128$~kbps, while the remaining tasks belong to the voice class with an arrival rate of $32$~kbps. In Fig.~\ref{fig:snapshot}, we can easily see how the agents and tasks can agree on a partition whereby a number of agents service a group of nearby tasks for data collection and transmission.  For the network of Fig.~\ref{fig:snapshot}, the tasks are distributed into three coalitions, two of which (coalitions $S_1$ and $S_3$) are served by a single collector-agent. In contrast, coalition $S_2$ is served by two collectors and one relay. The agents in coalition $2$ distributed their roles (relay or collector) depending on the achieved utility. For instance, for coalition $S_2$, having two collectors and one relay provides a utility of  $v(S_2)=52.25$ while having three collectors yields a utility of $v(S_2)=10.59$ , and having one collector and two relays yields a utility of $v(S_2)=45.19$. As a result, the case of two collectors and one relay maximizes the utility and is agreed upon between the players. Further, the coalitions in Fig.~\ref{fig:snapshot} are dynamic, in the sense that, within each coalition, the agents move from one task to the other, collecting and transmitting data to the receiver continuously. The order in which the agents visit the tasks, as indicated in Fig.~\ref{fig:snapshot}, is generated using a nearest neighbor solution for the traveling salesman problem as given by Property~\ref{prop:sales}.  For example, consider coalition $S_2$ in Fig.~\ref{fig:snapshot}. In this coalition, agents $3$ an $5$ act as a single collector and move from task $1$, to task $3$, to task $10$, and then back to task $1$ repeating these visits in a cyclic manner. Concurrently with the collectors movement, agent $2$ of coalition $S_2$, moves and positions itself at the middle of the line connecting the task being serviced by agents $3$ and $5$ to the central receiver. In other words, when the collectors are servicing task $1$ agent $2$ is at the middle of the line connecting task $1$ to the central receiver, subsequently when the collectors are servicing task $3$ agent $2$ takes position at the middle of the line connecting task $3$ to the central receiver and so on. Finally, note that, for all the coalitions in Fig.~\ref{fig:snapshot} one can verify that the minimum number of collectors, as per Proposition~\ref{prop:property2} is approximately $1$, (e.g., for coalition $S_2$, we have $|G_{S_2}|_{\textrm{min}}= \frac{9}{24}$ thus $1$ collector is a minimum), and, hence, this condition is easily satisfied by the coalition formation process.
\begin{figure}[!t]
\begin{center}
\includegraphics[width=80mm]{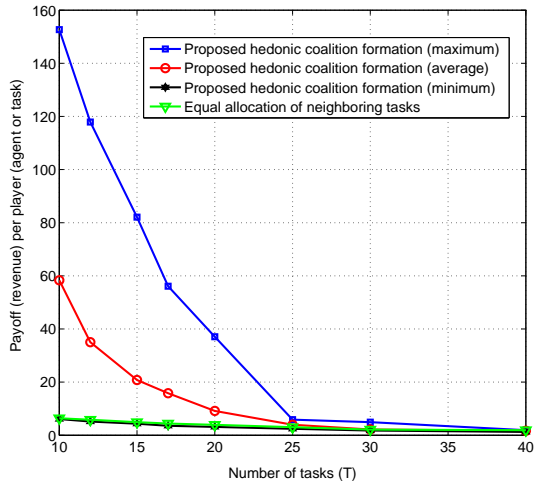}
\end{center}\vspace{-0.5cm}
\caption {Performance statistics, in terms of  maximum, average and minimum (over the order of play) player payoff, of the proposed algorithm compared to an algorithm that allocates the neighboring tasks equally among the agents as the number of tasks increases for $M=5$~agents.} \label{fig:perf}\vspace{-0.8cm}
\end{figure}

In Fig.~\ref{fig:perf}, we assess the performance of the proposed algorithm, in terms of the payoff (revenue) per player (agent or task) for a network having $M=5$~agents, as the number of tasks increases. The figure shows the statistics (averaged over the random positions of the tasks), in terms of maximum, average, and minimum over the random order of play. We compare the performance with an algorithm that assigns the tasks equally among the agents (i.e., an equal group of neighboring tasks are assigned for every agent). Fig.~\ref{fig:perf} shows that the performance of both algorithms is bound to decrease as the number of tasks increases. This is mainly due to the fact that, for networks having a larger number of tasks, the delay incurred per coalition, and, thus, per user increases. This increase in the delay is not only due to the increase in the number of tasks, but also to the increase in the distance that the agents need to travel within their corresponding coalitions (increase in switchover times).  In Fig.~\ref{fig:perf}, we note that the minimum payoff achieved by the proposed algorithm is comparable to that of the equal allocation. Hence, the performance of the proposed algorithm is clearly lower bounded by the equal allocation algorithm. However, Fig.~\ref{fig:perf}  shows  that the average and maximum payoff resulting by the proposed algorithm is significantly better than the equal allocation at all network sizes up to $T=25$~tasks. Albeit this performance improvement decreases with the increase in the number of tasks, the performance, in terms of average payoff per player, yielded by the proposed algorithm is no less than $30.26\%$ better than the equal allocation for up to $T=25$~tasks. Beyond $T=25$~tasks, Fig.~\ref{fig:perf} shows that the average and maximum performance of the proposed algorithm is comparable to that of the equal allocation, notably at $T=40$~tasks. The reduction in the performance gap between the two algorithms for large networks stems from the fact that, as more tasks exist in the network, for a fixed number of agents, the possibility of forming large coalitions, using the proposed algorithm is reduced, and, hence, the structure becomes closer to equal allocation.
\begin{figure}[!t]
\begin{center}
\includegraphics[width=83mm]{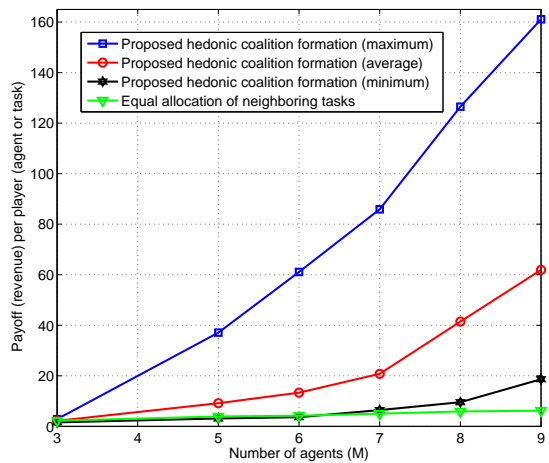}
\end{center}\vspace{-0.55cm}
\caption {Performance statistics, in terms of  maximum, average and minimum (over the order of play) player payoff, of the proposed algorithm compared to an algorithm that allocates the neighboring tasks equally among the agents as the number of agents increases for $T=20$~tasks.} \label{fig:perfage}\vspace{-0.7cm}
\end{figure}

In Fig.~\ref{fig:perfage}, we show the statistics (averaged over the random positions of the tasks), in terms of maximum, average, and minimum (over the random order of play) payoff per player for a network with $T=20$~tasks as the number of agents $M$ increases. The performance is once again compared with an algorithm that assigns the tasks equally among the agents (i.e., an equal group of neighboring tasks are assigned for every agent). Fig.~\ref{fig:perfage} shows that the performance of both algorithms increases as the number of agents increases. This is mainly due to the fact that when more agents are deployed, the tasks can be better serviced as the delay incurred per coalition decreases and the probability of successful transmission improves. For instance, as more agents enter the network, they can act as either collectors (for improving the delay) or relays (for improving the success probability). We note that, at $M=3$, the performance statistics of the proposed algorithm converge to the equal allocation algorithm since, for a small number of agents, the flexibility of forming coalitions is quite restricted and equal allocation is the most straightforward coalitional structure. Nonetheless, Fig.~\ref{fig:perfage} shows that, as $M$ increases, the performance of the proposed algorithm, in terms of maximum and average payoff achieved, becomes significantly larger than that of the equal allocation algorithm, and this performance advantage increases as more agents are deployed. Finally, Fig.~\ref{fig:perfage} also shows that the minimum performance of the proposed algorithm is comparable to the equal allocation algorithm for network with a small number of agents, but as the number of agents increases, the minimum performance of hedonic coalition formation is $29\%$ better than equal allocation at $M=7$, and this advantage increases further with $M$.
\begin{figure}[!t]
\begin{center}
\includegraphics[width=80mm]{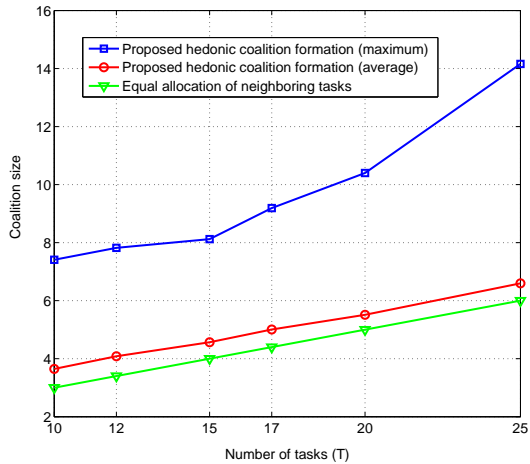}
\end{center}\vspace{-0.55cm}
\caption {Average and maximum (over order of play) coalition size yielded by the proposed algorithm and an algorithm that allocates the neighboring tasks equally among the agents, as a function of the number of tasks $T$ for a network of $M=5$ agents.} \label{fig:coalsize}\vspace{-0.8cm}
\end{figure}

In Fig.~\ref{fig:coalsize}, we show the average and maximum (over the random order of play) coalition size resulting from the proposed algorithm as the number of tasks $T$ increases, for a network of $M=5$~agents and arbitrarily deployed tasks. These results are averaged over the random positions of the tasks and are compared with the equal allocation algorithm. Fig.~\ref{fig:coalsize} shows that, as the number of tasks increases, the average coalition size for both algorithms increases. For the proposed algorithm, the maximum coalition size also increases with the number of tasks. This is an immediate result of the fact that, as the number of tasks increases, the probability of forming larger coalitions is higher and, hence, our proposed algorithm yields larger coalitions. Further, at all network sizes, the proposed algorithm yields coalitions that are relatively larger than the equal allocation algorithm. This result implies that, by allowing the players (agents and tasks) to selfishly select their coalitions, through the proposed algorithm, the players have an incentive to structure themselves in coalitions with average size lower bounded by the equal allocation. In a nutshell, through hedonic coalition formation, the resulting topology mainly consists of networks composed of a large number of small coalitions as demonstrated by the average coalition size. However, in a limited number of cases, the network topology can also be composed of a small number of large coalitions as highlighted by the maximum coalition size shown in Fig.~\ref{fig:coalsize}.
%1- Performance plot vs. tasks , vs. UAVs
%2- Average num of coalitions vs. num TAsks (average, min max)
%3- Beta
%4- Snapshot
%5- Split
%6- Join/leave tasks
%7- Mobility of tasks
%8- vs. SNR target
%9- System model
%10- Table

In Fig.~\ref{fig:speed},  we show, over a period of $5$ minutes, the frequency in terms of average switch operations per minute per player (agent or task) achieved for various velocities of the tasks in a mobile wireless network with $M=5$~agents and different number of tasks. As the velocity of the tasks increases,  the frequency of the switch operations
increases for both $T=10$ and $T=20$ due to the changes in the positions of the various tasks incurred by mobility. Fig.~\ref{fig:speed} shows that the case of $T=20$~tasks yields a frequency of switch operations significantly higher than the case of $T=10$~tasks. This result is interpreted by the fact that, as the number of tasks increases, the possibility of finding new partners as the tasks move increases significantly, hence yielding an increase in the topology variation as reflected by the number of switch operations. In summary, this figure shows that hedonic coalition formation allows the agents and the tasks to self-organize and adapt their topology to mobility, through adequate switch operations.
\begin{figure}[!t]
\begin{center}
\includegraphics[width=80mm]{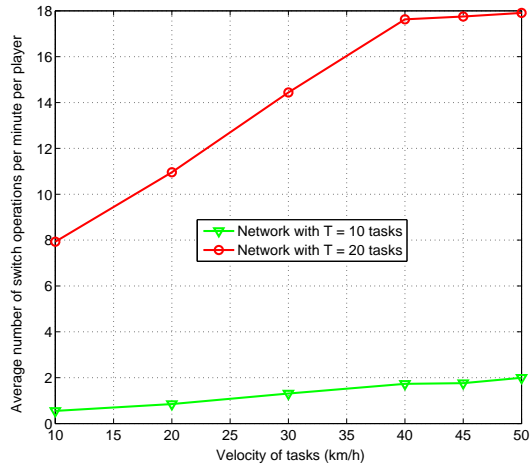}
\end{center}\vspace{-0.55cm}
\caption {Frequency of switch operations per minute per player achieved over a period of $5$ minutes for different tasks' velocities in a network having $M=5$~agents and different number of mobile tasks.} \label{fig:speed}\vspace{-0.8cm}
\end{figure}

The network's adaptation to mobility is further assessed in Fig.~\ref{fig:life}, where we show, over a period of $5$ minutes, the average coalition lifespan (in seconds) achieved for various velocities of the tasks in a mobile wireless network with $M=5$~agents and different number of tasks. The coalition lifespan is defined as the time (in seconds) during
which a coalition is present in the mobile network prior to accepting new members or breaking into smaller coalitions (due
to switch operations). Fig.~\ref{fig:life} shows that, as the velocity of the tasks increases,  the average lifespan of a coalition
decreases. This is due to the fact that, as mobility increases, the possibility of forming new coalitions
or splitting existing coalitions increases significantly. For example, for $T=20$, the coalition
lifespan drops from around $124$~seconds for a tasks' velocity of $10$~km/h to just under a minute as of $30$~km/h, and down to around $42$ seconds at $50$~km/h. Furthermore, Fig.~\ref{fig:life} shows that
as more tasks are present in the network, the coalition lifespan decreases. For instance, for any given velocity,
the lifespan of a coalition for a network with $T=10$~tasks is significantly larger than that of a coalition
in a network with $T=20$~tasks. This is a direct result of the fact that, for a given tasks' velocity, as more tasks are present in the network,
the players are able to find more partners to join with, and hence the lifespan of the coalitions becomes
shorter. In brief, Fig.~\ref{fig:life} provides an interesting assessment of the topology adaptation
aspect of the proposed algorithm through the process of forming new coalitions or breaking existing coalitions.

\begin{figure}[!t]
\begin{center}
\includegraphics[width=80mm]{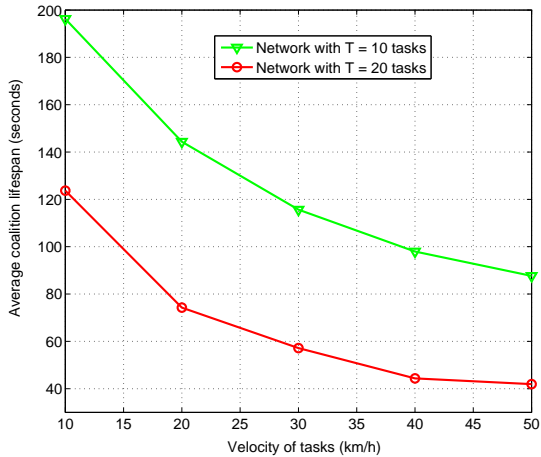}
\end{center}\vspace{-0.55cm}
\caption {Average coalition lifespan (in seconds) achieved over a period of $5$ minutes for different tasks' velocities in a network having $M=5$~agents and different number of mobile tasks.} \label{fig:life}\vspace{-0.4cm}
\end{figure}
Moreover, for further analysis of the self-adapting aspect of the proposed algorithm, we study the variations of the coalitional structure over time for a network where tasks are entering and leaving the network. For this purpose, in Fig.~\ref{fig:topadd}, we show the variations of the average (over the random positions of the tasks) number of players per coalition, i.e., the average coalition size, over a period of $10$ minutes, %(this period is followed by a period $\Psi$ during which Phase~3 of the algorithm occurs as per Table~\ref{tab:alg1}),
as new tasks join the network and/or existing tasks leave the network. The considered network in Fig.~\ref{fig:topadd} possesses $M=5$~agents and starts with $T=15$ tasks. The results are shown for different rates of change which is defined as the number of tasks that have either entered the network or left the network per minute. For example, a rate of change of $2$ tasks per minute indicates that either $2$ tasks enter the network every minute, $2$ tasks leave the network every minute, or $1$ tasks enters the network and another tasks leaves the network every minute (these cases may occur with equal probability). In this figure, we can see that, as time evolves, the structure of the network is changing, with new coalitions forming and other breaking as reflected by the change in coalitions size. Furthermore, we note that, as the rate of change increases, the changes in the topology increase. For instance, it is seen in Fig.~\ref{fig:topadd} that for a rate of change of $5$ tasks per minute, the variations in the coalition size are much larger than for the case of $2$ tasks per minute (which is almost constant for many periods of time). In summary, Fig.~\ref{fig:topadd} shows the network topology variations as tasks enter or leave the network. Note that, after the $10$ minutes have elapsed, the network re-enters in the Phase~III of the algorithm where data collection and transmission occurs.

\begin{figure}[!t]
\begin{center}
\includegraphics[width=80mm]{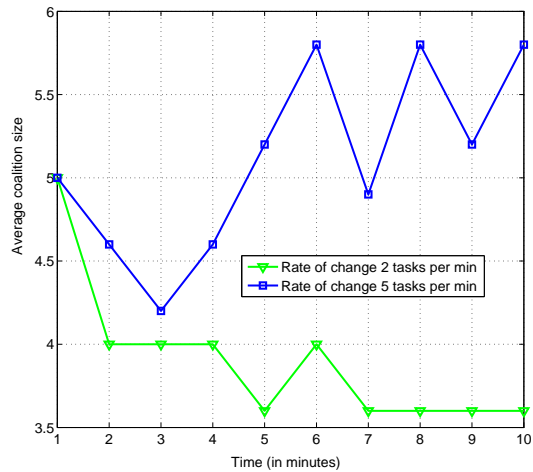}
\end{center}\vspace{-0.55cm}
\caption {Topology variation over time as new tasks join the network and existing tasks leave the network with different rates of tasks arrival/departure for a network starting with $T=15$ tasks and having $M=5$~agents.} \label{fig:topadd}\vspace{-0.6cm}
\end{figure}
In Fig.~\ref{fig:perfbeta}, we assess the performance of the proposed algorithm, in terms of the payoff (revenue) per player (agent or task) for a network having $M=5$~agents and $T=20$~tasks, as the throughput-delay tradeoff parameter $\beta$ increases. The figure shows the statistics, in terms of maximum, average, and minimum over the random order of play between the players. Fig.~\ref{fig:perfbeta} shows that, for small $\beta$, the performance of the proposed algorithm is comparable to the equal allocation algorithm and the payoffs are generally small. This result is due to the fact that, for small $\beta$, the tasks are highly delay sensitive, and the delay component of the utility governs the performance. Hence, for such tasks, the proposed algorithm yields a performance similar to equal allocation. However, as the tradeoff parameter $\beta$ increases, the maximum and average utility yielded by our proposed algorithm outperforms the equal allocation algorithm significantly. For instance, as of $\beta = 0.5$, hedonic coalition formation is highly desirable, and presents a performance improvement in terms of average payoff of around $19.56\%$ relative to the equal allocation algorithm (at $\beta =0.55$, the proposed algorithm has an average payoff of $ 0.55$ while equal allocation has an average payoff of $0.46$). This advantage increases with $\beta$. Note that, for all tradeoff parameters, the performance of the proposed algorithm, in terms of minimum (over order of play) payoff gained by a player is lower bounded by equal allocation and, in average, outperforms the equal allocation algorithm.\vspace{-0.3cm}
\begin{figure}[!t]
\begin{center}
\includegraphics[width=80mm]{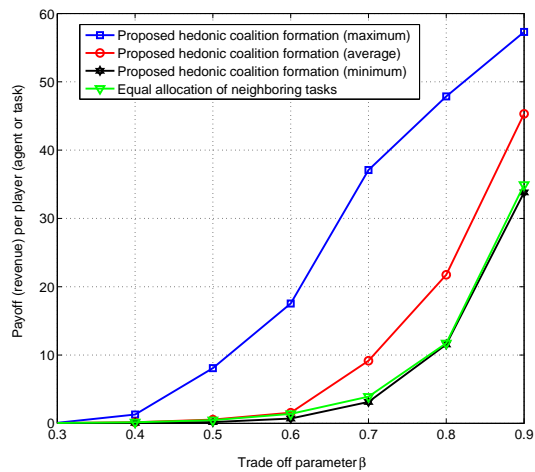}
\end{center}\vspace{-0.55cm}
\caption {Performance statistics, in terms of  maximum, average and minimum (over the order of play) player payoff, of the proposed algorithm compared to an algorithm that allocates the neighboring tasks equally among the agents as the throughput-delay tradeoff parameter $\beta$ increases for  $M=5$ agents and $T=20$~tasks.} \label{fig:perfbeta}\vspace{-0.8cm}
\end{figure}

\section{Conclusions}\label{sec:conc}\vspace{-0.2cm}
In this paper, we introduced a novel model for task allocation among a number of autonomous agents in a wireless communication network. In this model, a number of wireless agents are required to service several tasks, arbitrarily located in a given area. Each task represents a queue of packets that require collection and wireless transmission to a centralized receiver by the agents. The task allocation problem is modeled as a hedonic coalition formation game between the agents and the tasks that interact in order to form disjoint coalitions. Each formed coalition is mapped to a polling system which consists of a number of agents continuously collecting packets from a number of tasks. Within a coalition, the agents can act either as collectors that move between the different tasks present in the coalition for collecting the packet data, or relays for improving the wireless transmission of the data packets. For forming the coalitions, we introduce an algorithm that allows the players (tasks or agents) to join or leave the coalitions based on their preferences which capture the tradeoff between the effective throughput and the delay achieved by the coalition. We study the properties and characteristics of the proposed model, we show that the
proposed hedonic coalition formation algorithm always converges to a Nash-stable partition, and we
study how the proposed algorithm allows the agents and tasks to take distributed decisions for adapting
the network topology to environmental changes such as the deployment of new tasks, the removal of
existing tasks or the mobility of the tasks. Simulation results show how the proposed algorithm allows the agents and tasks to self-organize into independent coalitions, while improving the performance, in terms of average player (agent or task) payoff, of at least $30.26\%$  (for a network of $5$ agents with up to $25$~tasks)  relatively to a scheme that allocates nearby tasks equally among the agents. In a nutshell, by combining concepts from wireless networks, queueing theory and novel concepts from coalitional game theory, we proposed a new model for task allocation among autonomous agents in communication networks which is well suited for many practical applications such as data collection, data transmission, autonomous relaying, operation of message ferry (mobile base stations), surveillance, monitoring, or maintenance
in next-generation wireless networks.\vspace{-0.3cm}

\nocite{WS00}
\renewcommand{\baselinestretch}{0.85}
\bibliographystyle{IEEEtran}
\bibliography{references}\vspace{-1cm}

\end{document}